\documentclass[aps,prl,reprint,groupedaddress]{revtex4-2}

\usepackage[utf8]{inputenc}
\usepackage{graphicx} 
\usepackage{amscd,amsmath,amssymb,amsfonts,xspace,mathrsfs,mathtools,amsthm,xcolor}
\usepackage{tabularx}
\usepackage{subcaption}
\usepackage{hyperref}
\usepackage{soul}
\usepackage{bbold}
\usepackage{bbm}
\hypersetup{
    colorlinks,
    linkcolor={red!50!black},
    citecolor={blue!50!black},
    urlcolor={blue!80!black}
}
\usepackage{parskip}
\usepackage{color-edits}
\usepackage{float}
\addauthor{ja}{orange}
\addauthor{aw}{blue}
\addauthor{rl}{red}

\newcommand{\R}{\mathbb{R}}
\newcommand{\C}{\mathbb{C}}
\newcommand{\E}{\mathbb{E}}
\newcommand{\CP}{\mathbb{P}}

\newcommand{\id}{\mathbbm{1}}

\newcommand{\norm}[1]{\left\lVert#1\right\rVert}
\newcommand{\ess}{\mathrm{ess}}

\newtheorem{remark}{Remark}

\newtheorem{definition}{Definition}
\newtheorem{theorem}{Theorem}

\newtheorem{corollary}{Corollary}

\usepackage{thmtools}
\usepackage{thm-restate}

\usepackage{cleveref}

\begin{document}


\title{Geometric quantum control and the random Schrödinger equation}

\author{Rufus Lawrence}
\email{lawrejam@fel.cvut.cz}
\affiliation{%
Department of Computer Science,
 Czech Technical University in Prague, Czech Republic 
}%

\author{Ale\v{s} Wodecki}
\email{wodecki.ales@fel.cvut.cz}
\affiliation{%
Department of Computer Science,
 Czech Technical University in Prague, Czech Republic 
}%

\author{Johannes Aspman}
\email{johannes.aspman@cvut.cz}
\affiliation{%
Department of Computer Science,
 Czech Technical University in Prague, Czech Republic 
}%

\author{Lloren\c{c} Balada Gaggioli}%
 \email{llorenc.balada.gaggioli@fel.cvut.cz}
\affiliation{%
Department of Computer Science,
 Czech Technical University in Prague, Czech Republic 
}%
\affiliation{%
LAAS-CNRS, Université de Toulouse, Toulouse, France
}%

\author{Jakub Mare\v{c}ek}%
 \email{jakub.marecek@fel.cvut.cz}
\affiliation{%
Department of Computer Science,
 Czech Technical University in Prague, Czech Republic 
}%




\date{\today}

\begin{abstract}
Understanding and mitigating noise in quantum systems is a fundamental challenge in achieving scalable and fault-tolerant quantum computation. Error modelling for quantum systems can be formulated in many ways, some of which are very fundamental, but hard to analyse (evolution by general dynamical map) and others perhaps too simplistic to represent physical reality. In this paper, we present an intermediate approach, introducing the random Schrödinger equation, with a noise term given by a time-varying random Hermitian matrix as a means to model noisy quantum systems. We derive bounds on the error of the synthesised unitary in terms of bounds on the norm of the noise, and show that for certain noise processes these bounds are tight. We then show that in certain situations, minimising the error is equivalent to finding a geodesic on $SU(n)$ with respect to a Riemannian metric encoding the coupling between the control pulse and the noise process, thus connecting our work to the complexity geometry pioneered by Michael Nielsen.
\end{abstract}


\maketitle

\section{Introduction}

In order to achieve the low gate-error rates necessary for scalable fault-tolerant quantum computation, one needs to mitigate the effect of noise. In the context of quantum optimal control, mitigating the effect of noise on the system requires a model of the noise which faithfully represents the behaviour of the quantum system while remaining tractable mathematically and numerically. Some models of open (noisy) quantum systems are hard to analyse, while others are too simplistic to represent physical reality. 

The standard approach to modelling open quantum systems is to use the Gorini–Kossakowski–Sudarshan–Lindblad (GKSL) equation:

\begin{align}
    \frac{d\rho}{dt} = [H, \rho] + \sum_j\gamma_j \big( L_j \rho L_j^\dagger - \frac{1}{2}\{L_jL_j^\dagger, \rho\}\big).
\end{align}

Here, $H$ denotes the system Hamiltonian, the (non-Hermitian) operators $L_j$ are the \textit{jump operators} and the real numbers $\gamma_j$ denote the \textit{damping rates}. However, it is difficult to apply this approach to a specific quantum system, since in general, it is very difficult to write down the jump operators and damping rates for a given physical system from first principles. Although it is sometimes possible to  identify them experimentally, see for instance \cite{waqassurvey}, this approach also has its limitations. 

An alternative approach of representing 
the state of an open quantum system as a time-dependent random vector in a Hilbert space goes back at least to Belavkin \cite{Belavkin1990}. 
Subsequently, Gisin and Percival \cite{Gisin_1992} derive the \textit{stochastic Schrödinger equation} as an \textit{unraveling} of the GKSL equation. That is to say, given a GKSL equation for the density operator $\rho$, the authors construct an Itô-type stochastic differential equation (SDE) for the state $|\psi\rangle$ such that the expectation $\E[|\psi\rangle\langle\psi|]$ satisfies the original GKSL equation. A Schrödinger equation associated in such a way to the GKSL equation is called a stochastic Schrödinger equation. These types of equations have also been studied in the case of non-Markovian open quantum systems  \cite[e.g.]{STRUNZ2001237,RevModPhys.89.015001, chenudelcampo}. It should be noted that while the dynamics of the stochastic Schrödinger equation give rise to mixed states after taking an ensemble average over the probability space, the trajectories $|\psi(t, \omega)\rangle$ evolve unitarily in this approach. 

The purpose of this paper is twofold. We begin by showing that a large class of open quantum systems, sometimes called \textit{noisy quantum systems} \cite{xue2025traversingquantumcontrolrobustness, kosut2022robust, aspman2024robustquantumgatecomplexity}, can be modelled by a random ordinary differential equation (RODE) which we refer to as the random Schrödinger equation. Just like the stochastic Schrödinger equation, the dynamics of random Schrödinger equation also give rise to mixed states upon taking ensemble averages. We then use this formalism to study robust optimal quantum control. 

The random Schrödinger for the time-evolution operator $U_S(t, \cdot)$ is the following RODE:

\begin{align}\label{eq:full_evolution_equation}
\begin{split}
    \frac{d U_S(t, \cdot)}{dt} &= -i (H_{0,S}(t) + H_{1,S}(t, \cdot)) U_S(t, \cdot), \\  U_S(0, \cdot) &= \id. 
\end{split}
\end{align}

Here, $H_{0,S}$ is the deterministic (control) Hamiltonian, $H_{1,S}$ is an essentially bounded, Hermitian matrix-valued random process representing the noise, and $S$ denotes that the operators are considered in the Schrödinger picture. Just as for the stochastic Schrödinger equation, the trajectories of solutions evolve unitarily, and mixed states are only obtained after taking ensemble averages. The RODE approach to modelling noisy quantum systems is complementary to the Itô-type SDE modelling introduced in \cite{Gisin_1992}. In particular, while Itô calculus and the associated SDE yield the appropriate framework for modelling quantum systems subject to ``white" Hermitian Gaussian noise, the RODE approach lends itself more readily to modelling systems subject to time-correlated, bounded Hermitian noise processes. 

The control theory of open quantum systems presents many challenges, and is not well studied. While Dirr et al. \cite{dirr2009lie} have studied the controllability/reachability of density matrices subject to GKSL dynamics in abstract terms using techniques from differential geometry, not much is known about the precision- or fidelity-optimal control problem. The control theory of the stochastic Schrödinger equation has been studied in \cite{Kallush_2014} and \cite{muller_et_al}, where the authors derive bounds that are similar to the bounds derived in this paper. Recently, some progress has been made on the robust control problem in \cite{kosut2022robust, kosut2023robust, xue2025traversingquantumcontrolrobustness, altafini2004coherent, controlling_the_uncontrollable, Yang_2024, Nam_Nguyen_2024}, as well as in \cite{funckeberberich} in the context of robust quantum annealing. In these works, the noise model is somewhat more ad-hoc, consisting simply of a ``generic" error term in the Hamiltonian, and probabilistic questions are not fully addressed. The RODE approach introduced in this paper should be considered as a probabilistic improvement on these models, rather than as a substitute for the GKSL/stochastic Schrödinger approach. 

By introducing a coupling between the control $H_{0,S}$ and the noise $H_{1,S}$, we are able to use the random Schrödinger model to place the ideas of \cite{kosut2022robust, xue2025traversingquantumcontrolrobustness, xue2025traversingquantumcontrolrobustness, funckeberberich} on a more rigorous mathematical and probabilistic footing. In this paper, we provide a first step in this direction by showing that for certain classes of Hermitian noise processes, it makes sense to model a quantum system as an RODE. Following this, we provide ``worst case" bounds that apply to any such noise. Subsequently, we formulate a condition on the noise that makes this bound tight. These bounds allow us to discuss a natural geometric interaction between the control field and the noise, which  ultimately allows us to link our results to the seminal work of Nielsen et al. \cite{nielsen2006geometric, nielsen2006optimal, nielsen2006quantum, dowling2008geometry} on the geometry of quantum computation. In particular, we show that for a natural class of couplings between control and noise,  the fidelity-optimal control problem reduces to the problem of finding geodesics with respect to an appropriately defined \textit{noise metric}. These noise metrics are analogous to the \textit{complexity metrics} introduced by Nielsen and his coauthors.

\subsection{Summary of results}

In this section, we summarize the main results of the paper, giving physical intuition and sketching proofs, where possible. The full proofs are given in the following sections of the paper. 

Our first contribution is a bound on the error induced by the noise term $H_{1,S}$. Let $U_{0,S}$ be the operator that satisfies the noiseless equation
\begin{align}
\begin{split}
\frac{dU_{0,S}(t)}{dt} &=-iH_{0,S}(t)U_{0,S}(t) \\U_{0,S}(0) &=\id.    
\end{split}
\end{align}
The following theorem quantifies the difference between the noiseless evolution given by $U_{0,S}$ and the noisy evolution $U_S$ (see \eqref{eq:full_evolution_equation}).

\vskip 6pt
\begin{restatable}[]{thm}{basicerror}
\label{thm:basicerror}
If $\ess\,\sup_{t, \omega} \norm{H_{1,S}(t,\omega)} \leq K$, the worst case error induced by the noise grows linearly in time. That is, $\ess\,\sup_\omega \norm{U_S(t, \omega) - U_{0,S}(t)} \leq Kt$ with probability $1$. 
\end{restatable}
\begin{proof}
(Sketch) Let $U_{0,S}(t)$ be the propagator associated to $H_{0,S}$, transform \eqref{eq:full_evolution_equation} into the interaction picture via $U_{0,S}$ and integrate. The result then follows by applying standard matrix norm inequalities.
\end{proof}

The subsequent theorem provides a geometric refinement of the previous theorem.

\vskip 6pt
\begin{restatable}[]{thm}{geometricerror}
\label{thm:geometricerror}
If $\Lambda(t) := \ess \, \sup_\omega \norm{H_{1,S}(t, \omega)}_F$, we have the following inequality:
\begin{align}
\begin{split}
    d_F(U_I(t, \omega), U_{0,S}(t)) &\leq \int_0^t \norm{H_{1,S}(s, \omega)}_Fds \\ &\leq \int_0^t \Lambda(s) ds.
\end{split}
\end{align}
\end{restatable}
\begin{proof}
(Sketch) Similar to the previous proof.    
\end{proof}

Here, $\norm{\cdot}_F$ denotes the Frobenius norm, and $d_F(\cdot, \cdot)$ the metric (distance function) on $SU(n)$ induced by the Frobenius norm. We introduce a condition on the noise that makes the bound  tight, which we call the \textit{worst-case noise condition} (WCNC), see \cref{def:WCNC1}. As the name suggests, the WCNC is a condition that guarantees that the noise is aligned in the worst possible direction with maximal intensity with positive probability. For noise processes satisfying the WCNC the bound in \cref{thm:geometricerror} is tight in a probabilistic sense:

\vskip 6pt
\begin{restatable}[]{thm}{tightness}
\label{thm:tightness}
For noise processes $H_{1,S}$ satisfying the \textit{worst-case noise condition} (see \cref{def:WCNC1}), we have that:
\begin{align}
    \CP \left[ \int_0^t \Lambda(s) ds - d_F(U_I(t, \omega), \id) < \epsilon \right] > 0.
\end{align}
\end{restatable}

\begin{proof}
(sketch) In general, we expect the noise to exhibit a degree of self-cancellation, see remark \ref{rk:uncorrelated}. However, this can fail: if, in the interaction picture the noise process is always oriented in the same axis and close to its bound we show that the solution approaches its bound. The proof is technical and relies on a secondary interaction picture transformation.
\end{proof}

We then provide a general result that gives conditions on the noise $H_{1,S}(t, \omega)$ under which the \textit{worst-case noise condition} of \cref{def:WCNC1} may be satisfied.

\vskip 6pt
\begin{restatable}[]{thm}{markov}
\label{thm:markov}
Let $\mathcal{H}_E$ denote the space of Hermitian matrices with Frobenius norm less than $E$, and let $H_{1,S}(t, \cdot)$ be a stochastic process with almost surely continuous sample paths and let $H_{0,S}(t) \in \mathcal{H}_E$ for some $E \in \R$. Suppose further that $H_{1,S}(t, \cdot)$ satisfies the strong Markov property, and that the associated measure on $C^0([0,T], \mathcal{H}_E)$ is positive on cylindrical sets. Then $H_{1,S}(t, \cdot)$ satisfies the worst-case noise condition with respect to $H_{0,S}$.
\end{restatable}
\begin{proof}
(sketch) The Markov property combined with the positivity condition imply that for any sufficiently regular path $\gamma$ in the space of Hermitian matrices, $\CP(\norm{H_{1,S}(t, \omega) - \gamma(t)} < \epsilon \; \forall t)$ is positive. The result then follows. 
\end{proof}

This theorem allows us to conclude that many well known stochastic processes satisfy the WCNC. In particular, if $X(t, \cdot)$ is an unbounded Markov process (e.g. a Gaussian Markov process), and $f$ is a continuous bounded function, then $H_{0,S}(t, \cdot) = f(X(t, \cdot))$ satisfies the WCNC for any $H_0(t)$ whose image lies in $\mathcal{H}_E$. See remark \ref{rk:nonemptycondition} for more detail.

We then introduce the notion of a control-noise coupling, which quantifies the relationship between the control field and the noise. In particular, we consider control-noise couplings induced by a Riemannian \textit{noise metric} denoted by $g_\Lambda$. Noise metrics are analogous to the complexity metrics introduced by Nielsen et. al. \cite{nielsen2006geometric, nielsen2006optimal, nielsen2006quantum, dowling2008geometry}. For these couplings, solving the robust control problem (see \cref{def:robustcontrol}) is equivalent to finding a $g_\Lambda$-geodesic between $\id$ and $V$, in a situation analogous to \cite{nielsen2006geometric, nielsen2006optimal, nielsen2006quantum, dowling2008geometry}.

\vskip 6pt
\begin{restatable}[]{thm}{geodesic}
\label{thm:geodesic}
Suppose that $H_{1,S}(t, \cdot)$ satisfies the control-dependent WCNC (see \cref{def:WCNC3}) with respect to all $H_{0,S}(t) \in \mathcal{H}_E$, and suppose $\Lambda(t) = \sqrt{g_\Lambda (H _{0,S}(t), H_{0,S}(t))}$. Then solving the robust optimal control problem of synthesising a unitary target $V$ with respect to the cost function $c(U_S(T), H_{0,S}(\cdot), T, V) = \ess\,\sup_\omega \norm{U_S(T, \omega)-V}$
is equivalent to finding a length-minimising $g_\Lambda$-geodesic on $SU(n)$.
\end{restatable}

\begin{proof}
    (Sketch) Plugging $\Lambda(t) = \sqrt{g_\Lambda (H _{0,S}(t), H_{0,S}(t))}$ into theorem \ref{thm:geometricerror} we see that the cost function of a given control is bounded by the length of the path traced out by the control in $SU(n)$, measured with respect to $g_\Lambda$. By theorem \ref{thm:tightness}, minimising the length of this path (i.e. finding a geodesic) minimises the cost function. 
\end{proof}

\section{Random Ordinary Differential Equations}

In this section, we review the basic theory of RODEs, following \cite{HanandKloeden, strand1970random}. See \cref{rk:ito-RODE} and \cref{rk:uncorrelated} for a comparison of RODEs and SDEs. 

\begin{definition}\label{def:RODEmain}
    Let $\eta:[0,T] \times \Omega \rightarrow \R^m$ be an $\R^m$-valued random process defined on a probability space $(\Omega, \mathcal{F}, \CP)$, and let $f:[0,T] \times \R^n \times \R^m\rightarrow \R^n$ be continuous. A random differential equation is an equation of the form

    \begin{align}\label{eq:RODEmain}
        \frac{dx(t, \cdot)}{dt} = f(t, x(t, \cdot),\eta(t, \cdot))\,; \quad x(0, \cdot) = x_0. 
    \end{align}
    We say that \eqref{eq:RODEmain} is linear if $f$ is linear in the state space (second argument). It follows that any linear RODE can be written as
    \begin{align}\label{eq:linearRODE}
        \frac{dx(t, \cdot)}{dt} = A(t, \cdot)x(t, \cdot); \quad x(0, \cdot) = x_0. 
    \end{align}
    Where $A: [0,T] \times \Omega \rightarrow \mathrm{Mat}_n(\R)$ is an $n\times n$ matrix-valued random process. 
\end{definition}

\begin{remark}
Our definition of an RODE is more restrictive than the definitions used by other authors, notably \cite{strand1970random}. Other authors define RODEs using the concept of a random function, which allows $f$ to depend directly on $\omega \in \Omega$. But this requires a significantly more technical treatment in the non-linear case. 

\end{remark}

There are several different ways to interpret the derivative on the left-hand side of an RODE, each giving rise to a different notion of what it means for a random process $x(t, \omega)$ to be a solution. 

\begin{definition}
We say that a random process $x(t, \omega)$ is a sample path (SP) solution on $\lbrack 0, T \rbrack$ if for almost all $\omega \in \Omega$, $x(\cdot, \omega)$ is absolutely continuous (in $t$),

\begin{align}\label{eq:SPRandom}
    \frac{\partial x(t, \omega)}{\partial t} = f(t,x(t,\omega), \eta(t, \omega))
\end{align}

for almost all $t \in \lbrack 0, T \rbrack$ and $x(0, \omega) = x_0(\omega)$. 
\end{definition}

\begin{remark}\label{rk:measurability}There is a slight subtlety to this definition. We have assumed that any SP solution is a priori a random process, i.e. that the mapping $x(t, \cdot):\Omega \rightarrow \R^n$ is measurable for all $t$. It turns out that this is not necessary. If we can solve the SP equation \eqref{eq:SPRandom} for almost all $\omega \in \Omega$, then the mapping

\begin{align*}
    x(t, \cdot): \Omega &\rightarrow \R^n \\
                \omega &\mapsto x(t, \omega)
\end{align*}
is defined almost everywhere, and moreover is measurable. For a proof see 2.1.2 in \cite{HanandKloeden}. 
\end{remark}
\begin{definition}
We say that $x$ is a mean-square ($L^2$) solution if $x$ is an $L^2$ random process (i.e. if $x(t, \cdot) \in L^2(\Omega)$ for all $t$), $x$ is strongly absolutely continuous as a function $x: [0,T] \rightarrow L^2(\Omega)$ and 

\begin{equation}\label{eq:L2_deriv}
\left[\frac{x(t+h,\cdot)-x(t,\cdot)}{h}\right]\overset{h\rightarrow0}{\longrightarrow}f(t,x(t,\cdot),\eta(t,\cdot)) \quad \text{in } L^{2}\left(\Omega\right)
\end{equation}

for almost all $t \in [0,T]$, and $x(0, \cdot) = x_0(\cdot)$ in $L^2(\Omega)$.

If $x(t, \cdot)$ is an $L^2$-process, we denote by $\hat{x}(t)$ its equivalence class in $L^2(\Omega)$. Clearly, $\hat{x}:[0,T] \rightarrow L^2(\Omega)$.
\end{definition}

Due to the absolute continuity assumptions, both the property of being an SP solution and the property of being an $L^2$ solution have equivalent integral formulations:

In the first case, $x(t, \cdot)$ is a SP solution if and only if 

\begin{equation}
    x(t, \omega) = x_0(\omega) + \int_0^t f(s, x(s, \omega), \eta(s, \omega)) ds
\end{equation}

for almost all $\omega \in \Omega$ and for almost all $[0,T]$. The integral above is just the usual Lebesgue integral of the trajectory indexed by $\omega$.

Similarly, an $L^2$ process $x(t, \cdot)$ is an $L^2$ solution if and only if

\begin{equation}
    x(t, \cdot) = x_0 + \int_0^t f(s, x(s, \cdot), \eta(t, \omega)) ds
\end{equation}

for all $t \in T$. Here, the integral denotes the \textit{Bochner integral}.

The relationship between sample path and $L^2$-solutions is less obvious. We have the following theorem:

\begin{theorem}\label{thm:SP-L2}
Suppose $\hat{x}(t)$ is an $L^2$-solution. Then there exists a jointly measurable map $x : \lbrack 0,T \rbrack \times \Omega \rightarrow \R^n$, equivalent to $\hat{x}$ in $L^2(\Omega)$, which is a SP solution.

Conversely, suppose $x : [0, T] \times \Omega \rightarrow \R^n$ is a SP solution. Then its equivalence class in $L^2(\Omega)$, $\hat{x}: [0,T] \rightarrow L^2 (\Omega)$ is an $L^2$ solution if furthermore
\begin{equation}
    \int_a^b \norm{f(s, \cdot, x(s, \cdot))}_{L^2} ds < \infty.
\end{equation}
\end{theorem}

In light of this theorem, as well as remark \ref{rk:measurability}, we will from now on move freely between the SP and $L^2$ notions of solutions wherever it is convenient and allowed. 

The main existence theorem for RODEs was proven in \cite{strand1970random} (Theorem 5). We reproduce it here for completeness:

\begin{theorem}\label{thm:RODEPicard}
    Let $\tilde{f}: [0,T] \times L^p(\Omega) \rightarrow L^p(\Omega)$ be the mapping induced by $f$. 

\begin{equation}\label{eq:LpLipschitz}
    \norm{\tilde{f}(t, x)-\tilde{f}(t, y)}_{L^p} \leq k(t) \norm{x - y}_{L^p}
\end{equation}
with $\int_a^b|k(t)|dt < \infty$, then the RODE

\begin{equation}
    \frac{dx(t, \cdot)}{dt} = f(t, x(t, \eta(t, \cdot)); \qquad x(0, \cdot) = x_0
\end{equation}
has a unique $L^p$-solution.
\end{theorem}

\begin{corollary}\label{cor:essbounded}
The linear RODE
\begin{align}
     \frac{dx(t, \cdot)}{dt} = A(t, \cdot)x(t,\cdot); \qquad x(0, \cdot) = x_0(\cdot)
\end{align}
admits a solution if $A(t,\cdot)$ is essentially bounded. By this we mean that there exists an $M$ such that for all $t \in [0,1]$, $\CP(\norm{A(t, \cdot)} \leq M) = 1$, where $\norm{\cdot}$ denotes any matrix norm.
\end{corollary}

We conclude this section with some general remarks.

\begin{remark}
By identifying $\C^n$ with $\R^{2n}$ we can treat complex-valued RODEs on the same footing as real valued RODEs.
\end{remark}

\begin{remark}\label{rk:compactness}
Corollary \ref{cor:essbounded} suggests that, if we want to obtain solutions to a linear RODE with finite moments up to a given order, we should restrict ourselves to studying RODEs with bounded right-hand side. However, there is one case where it is possible and indeed natural to consider unbounded linear RODEs. If $A(t,\omega) \in \mathfrak{g} \subset \mathrm{Mat}_n(\C)$, where $\mathfrak{g}$ is a Lie subalgebra of $\mathrm{Mat}_n(\C)$ such that the associated Lie group $G \subset GL_n(\C)$ is compact, and $A(t, \omega)$ is continuous in $t$ for almost all $\omega$, then the above RODE admits solutions with finite moments. The proof follows from using the escape lemma to solve the SP problem for each $\omega$, measurability then follows from remark \ref{rk:measurability} and finiteness of moments is evident since $x(t, \omega) \in G$ for all $t$ for almost all $\omega$, and $G$ is compact. 
\end{remark}

\begin{remark}\label{rk:ito-RODE}
Certain types of RODEs can be reinterpreted as an Itô-type SDE: if $W(t, \cdot)$ is the classical Wiener process, and $\mu$ and $\sigma$ are integrable functions, the RODE 

\begin{align}
    \frac{dx(t, \cdot)}{dt}=\mu(x(t,\cdot), t) + \sigma(x(t, \cdot), t)\frac{dW(t, \cdot)}{dt}
\end{align}

is equivalent to the Itô stochastic differential equation

\begin{align}
    dx(t, \cdot)=\mu(x(t,\cdot), t)dt + \sigma(x(t, \cdot), t)dW(t, \cdot).
\end{align}

Of course, $\frac{dW(t, \cdot)}{dt}$ only exists in the sense of distributions. 

\end{remark}

\begin{remark}\label{rk:uncorrelated}
Uncorrelated centered noise dose not affect the dynamics. Consider the RODE

\begin{align}\label{eq:uncorrelated}
\frac{dx(t, \cdot)}{dt} = A(t, \cdot)x(t,\cdot); \qquad x(0, \cdot) = x_0(\cdot)
\end{align}

and suppose that $\E[A(t, \cdot)] = 0$, and $\E[A_{ij}(t)A_{kl}(s)] = 0$ for all $s \neq t$. Then the only $L^2$ solution is given by $x(t) = x_0$ for all $t$. 

First, note that if $B_{ij}(t) = \int_0^t A_{ij}(s)ds$, then $B_{ij}(t) = 0$ in $L^2(\Omega)$ for all $t$, since 

\begin{align}\label{eq:Bochneruncorrelated}
\begin{split}
    \E[B_{ij}(t)^2] = \E[\int_0^tA_{ij}(s_1)ds_1\int_0^tA_{ij}(s_2)ds_2] \\ = \int_0^t \int_0^t\E[A_{ij}(s_1)A_{ij}(s_2)]ds_1ds_2 = 0.
\end{split}
\end{align}

Here, we have used Fubini's theorem to interchange integration with respect to $s_i$ and $\E$, and the fact that, by assumption, $\E[A_{ij}(s_1)A_{ij}(s_2)]$ is supported only on the diagonal $\{s_1 = s_2\}$, which has measure 0. 

Now, using equivalence of $L^2$ and SP solutions (\cref{thm:SP-L2}), we can solve \cref{eq:uncorrelated} using the Magnus expansion:

\begin{align}
    x(t, \omega) = \exp(\sum_{k=1}^\infty M_k)
\end{align}

with

\begin{align}
    M_1 &= \int_0^tA(t_1)dt_1 \\
    M_2 &= \frac{1}{2}\int_0^t \int_0^{t_1} [A(t_1),A(t_2)]dt_1dt_2 \\
        & \vdots
\end{align}

\cref{eq:Bochneruncorrelated} shows that $M_1$ vanishes for all $t$, and a similar argument shows that all the higher order terms vanish too. 

There is however one type of uncorrelated noise that does affect the dynamics, namely the Gaussian white noise of \cref{rk:ito-RODE}, which of course is not an $L^2$-process. Making sense of this rigorously is difficult, requiring techniques from Malliavin calculus. Setting aside technical details, if we denote $G(t, \cdot) = \frac{dW}{dt}(t, \cdot)$, we have that $\E[G(t),G(s)] = \delta(t-s)$. When we try to apply the same argument as in \cref{eq:Bochneruncorrelated}, we instead find that

\begin{align}
\begin{split}
    \E[W(t, \cdot),W(s, \cdot)] &= \int_0^t \int_0^s \delta(t_1-s_1)dt_1ds_1 \\ &= \min(t,s)
\end{split}
\end{align}

as expected. This illustrates the main difference between SDEs and RODEs:

\begin{itemize}
    \item RODEs are suitable for modelling systems subject to bounded noise, since otherwise we cannot in general expect that a solution even exists. If we want to have non-trivial stochastic behaviour, we need the noise to exhibit some degree of correlation at different times.
    \item SDEs are suitable for modelling systems subject to uncorrelated noise. However, in order for the stochastics to be non-trivial, the noise process must live in a more exotic class of stochastic process than the class of $L^2$-processes. Another drawback of SDEs is that, unlike RODEs, they are not amenable to path-wise analysis. 
\end{itemize}

\end{remark}

\subsection{Robust control for RODEs}\label{sec:robust}

Just as there is a notion of control theory for ODEs, there is a natural notion of \textit{robust control} theory for RODEs. Suppose we have an RODE

\begin{equation}
\begin{split}
    \frac{dx(t, \cdot)}{dt} &= f(a_1(t), \dots, a_k(t), x(t, \cdot), \eta(a_1(t), \dots, a_k(t),t, \cdot)) \\ x(0, \cdot) &= x_0 \in L^2(\Omega, \R^n)
\end{split}
\end{equation}

where the $a_j(t)$ are real-valued control functions. Note that even for complex RODEs, we can still assume that the control functions are real valued, by splitting into real and imaginary parts if necessary. Denote by $A$ the function space of admissible control functions, i.e. $a := (a_1(\cdot), \dots, a_k(\cdot)) \in A$. Note also that the noise process $\eta$ can depend on the control functions $a_j(t)$. We refer to this dependency as a \textit{control-noise} coupling.

\begin{definition}\label{def:robustcontrol}
    Given some vector $y$ in $\R^n$, the objective of the robust control problem is to find control functions $a_j(t)$ and a stopping time $T \in \R$ minimising the \textit{cost function} 
\begin{align}
    c(x(T, \cdot), a_j(\cdot), T, y).
\end{align}
\end{definition}

We will mostly restrict our attention to cost functions that depend only on the random variable $x(T, \cdot)$ and the target $y$. We also restrict our attention to the situation where $\E[x(T, \cdot)] = y$, and we assume that we can find a control $a \in A$ such that this occurs. We refer to this last condition as \textit{operator controllability in expectation}.

\section{The Random Schrödinger Equation}

The RODE we will study in this paper is the \textit{random Schrödinger equation} for the time evolution operator:

\begin{align}
\begin{split}
    \frac{d U_S(t, \cdot)}{dt} &= -i (H_{0,S}(t) + H_{1,S}(t, \cdot)) U_S(t, \cdot) \\ U_S(0, \cdot) &= \id 
\end{split}
\end{align}

where the subscript $S$ denotes that we are working in the Schrödinger picture. Here, $H_{0,S}(t)$ is the deterministic part of the Hamiltonian, while $H_{1,S}(t,\cdot)$ is a Hermitian matrix-valued random process. In what follows, we will assume that $H_{1,S}$ is \textit{essentially bounded}, although in view of remark \ref{rk:compactness}, this is not necessary to guarantee the existence of a solution. We may further assume that $H_{1,S}$ is centered, since we can always absorb the drift term into $H_{0,S}$. 

As expected, there is also random Schrödinger equation for (random) states:

\begin{align}
\begin{split}
    \frac{d |\varphi(t, \cdot)\rangle}{dt} &= -i (H_{0,S}(t) + H_{1,S}(t, \cdot)) |\varphi(t, \cdot)\rangle \\ |\varphi(0, \cdot)\rangle &= |\varphi_0(\cdot)\rangle,
\end{split}
\end{align}

and the usual relation between states and operators holds pathwise: $|\varphi(t, \omega)\rangle = U_S(t, \omega)|\varphi(\omega, 0)\rangle$. The following figure shows plots of solutions to the random Schrödinger equation for states, visualised on the Bloch sphere. 

\begin{figure}[H]
  \centering

  \begin{minipage}{0.25\textwidth}
    \includegraphics[width=\linewidth]{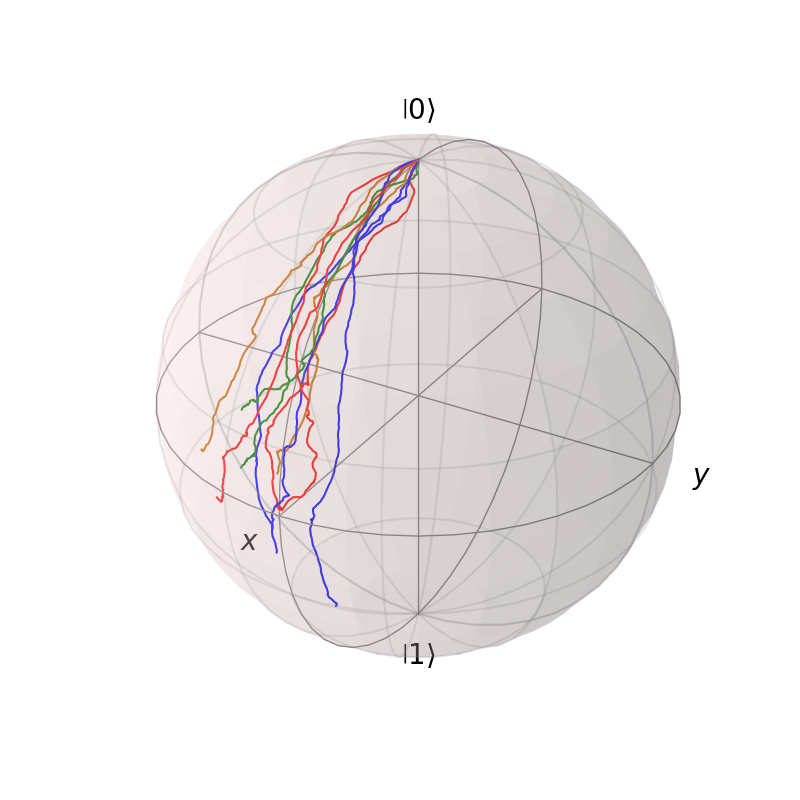}
  \end{minipage}
  \begin{minipage}{0.25\textwidth}
    \includegraphics[width=\linewidth]{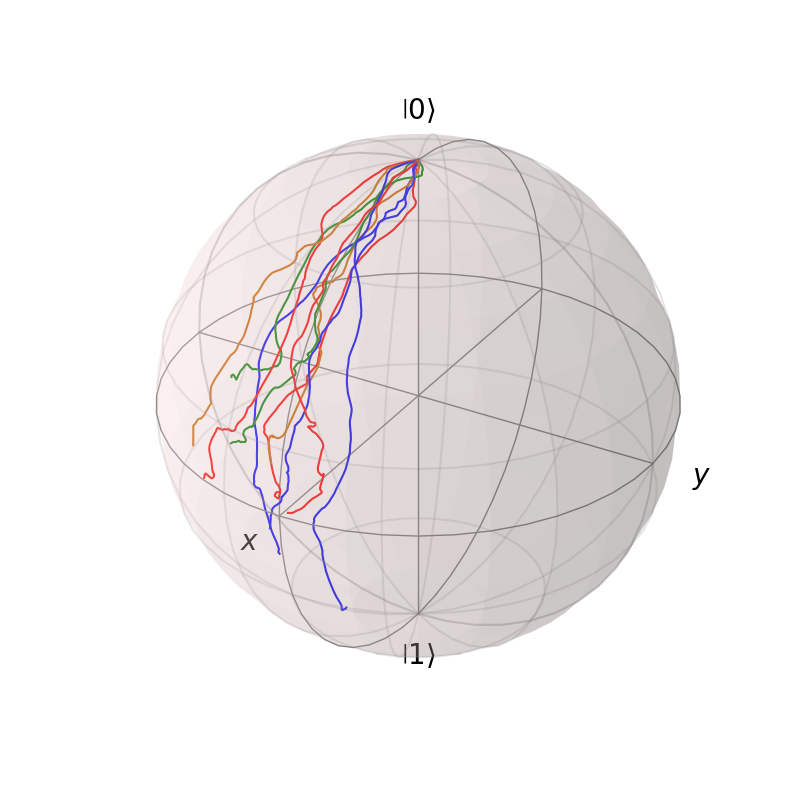}
  \end{minipage}
  \begin{minipage}{0.25\textwidth}
    \includegraphics[width=\linewidth]{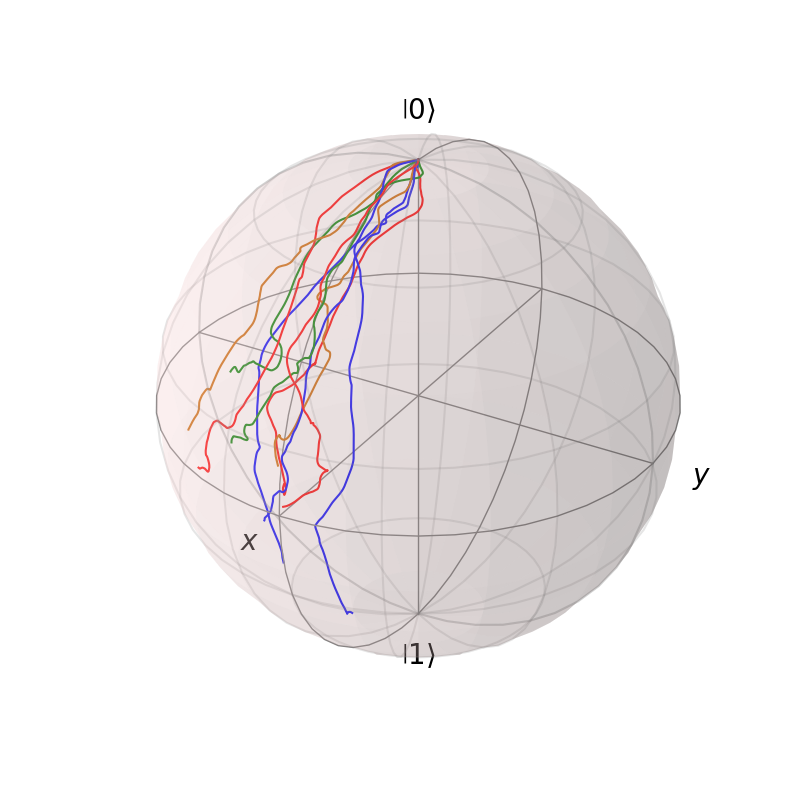}
  \end{minipage}

  \caption{Plots of 10 pathwise solutions to the random Schrödinger equation for a single qubit with $H_{0,S}(t) = cos(t)\sigma_y$ and $H_{1,S}(t,\omega) = \frac{1}{2} M_\nu(t, \cdot)$, where $M_\nu (t, \cdot)$ is a Gaussian process with Matérn covariance. The parameter $\nu$ controls the smoothness of sample paths. We take to $\nu$ to be $0.2, 0.6$ and $1.0$ from left to right. A Gaussian process with Matérn covariance is in general $\lceil \nu \rceil -1$ times $L^2$-differentiable. Trajectories were sampled using sci-kit learn and plotted using QuTiP \cite{qutip}.}
  \label{fig:trajectories_bloch}
\end{figure}

Before discussing the theoretical guarantees for systems modelled by the random Schrödinger equation, we discuss the relationship between the random Schrödinger equation and other more traditional noise models.

\subsection{Mixed state behaviour}

Since $H_{0,S}(t) + H_{1,S}(t, \omega)$ is Hermitian, the sample trajectories $U_S(t, \omega)$ are always unitary. However, by considering ensemble averages over the probability space $\Omega$, we obtain the mixed-state behaviour that we expect open quantum systems to exhibit. 

Consider a (normalised) random state $\psi \in L^2(\Omega, \mathcal{H})$. From this we can form a random density operator

\begin{align}
\begin{split}
    \rho: \Omega &\rightarrow \mathcal{D}(\mathcal{H}) := \{\rho \in \mathcal{B}(\mathcal{H}), \, \rho =\rho^\dagger, \, \rho \succeq 0, \mathrm{Tr}(\rho) = 1\}  \\
    \omega &\mapsto |\psi(\omega)\rangle\langle\psi(\omega)|.
\end{split}
\end{align}

For each $\omega \in \Omega$, $\rho(\omega)$ is a pure state, however, taking an ensemble average, we can obtain a matrix $\hat{\rho} = \E_\omega(\rho(\omega))$. 

Now, $\hat{\rho}$ is in fact a density matrix, since Tr$(\hat{\rho}) = 1$, and it is in general a mixed state. To see this, note that $\rho(\omega)$ is supported on the space of pure states,

\begin{align}
    \mathcal{D}(\mathcal{H}) \cap \{\mathrm{Tr}\rho^2 = 1\}.
\end{align}

We know that the expectation $\E(\rho)$ is contained in the convex hull of the support. Since the space of pure states contains no non-trivial convex subsets, it is clear that $\E(\rho)$ is a mixed state as soon as $\rho(\omega)$ is non-constant. 

\subsection{Generating standard error processes}

Despite its apparent simplicity, the random Schrödinger equation captures most ``standard" quantum errors considered in quantum information theory and quantum error correction.

Consider first the case of a bit flip channel on a single qubit. This can be described in terms of the Kraus maps $K_1 = \sqrt{p}X$, $K_1 = \sqrt{1-p}\,\mathbb{1}$, for some parameter $p \in [0,1]$. It is not difficult to write down a random Hamiltonian that generates this channel upon taking ensemble averages. If we set

\begin{align}
\begin{split}
    H_{0,S}(t, \omega) &= 
    \begin{pmatrix}
        0 & \frac{\pi}{2} \\
        -\frac{\pi}{2} & 0
    \end{pmatrix} \quad \text{with probability}\, p \\ 
    H_{0,S}(t, \omega) &= \begin{pmatrix}
        0 & 0 \\
        0 & 0
    \end{pmatrix} \quad \text{with probability}\, 1-p,
\end{split}
\end{align}

then evolving for one unit of time generates the bit flip channel. Of course this choice is non-unique. 

Indeed, the same trick works for any mixed unitary error admitting Kraus maps $\{p_jU_j\}$ with $\sum_jp_j = 1$. We can choose Hamiltonians $H_j(t)$ generating $U_j$, and let $H_{1,S}(t, \omega)$ be a random whose trajectory is given by $H_j(t)$ with probability $p_j$. Again, this choice is (very) non-unique. Furthermore, one could generate arbitrary quantum channels by evolving via non-Hermitian noise processes, as for instance in \cite{delcampo}. However, this is beyond the scope of this paper. 

Another common error encountered in the literature is the so-called $\frac{1}{f}$ noise. This fits naturally into our framework. For instance, in \cite{oneoverf}, the authors consider the random Hamiltonian

\begin{align}
    \hat H_q = \frac{1}{2}\mathbf{B}\cdot\vec{\sigma}.
\end{align}

Here, $\mathbf{B}$ denotes the magnetic field, and $\vec{\sigma}$ is the Pauli-vector. The authors model the  magnetic field as

\begin{align}
    \mathbf{B} = \mathbf{B}_0(t) + \mathbf{b}(t, \omega),
\end{align}

where $\mathbf{b}(t, \omega)$ is a stochastic process, which coincides with the random Schrödinger formalism so long as $\mathbf{b}(t, \omega)$ is essentially bounded. 

\subsection{Basic error estimate}

In this section, $\norm{\cdot}$ may be taken to denote the operator norm induced by the Hilbert norm on $\mathcal{H}$, although exactly the same results holds for the Frobenius norm $\norm{\cdot}_F$. The first natural question to ask is how the noise term affects the solution. In other words, if we denote by $U_{0,S}(t)$ the propagator associated to $H_{0,S}(t)$, how does the random variable $\norm{U_S(t, \cdot) - U_{0,S}(t)}$ evolve over time? 

The random-variable $\norm{U_S(t, \cdot) - U_{0,S}(t)}$ can be thought of as the error induced by perturbing the ideal control Hamiltonian $H_{0,S}$ by the additive noise $H_{1,S}(t, \cdot)$. We have the following theorem:

\basicerror*

\begin{remark}
This is much better than we can expect using standard results such as Gronwall's inequality. This is because the dynamics are constrained to $SU(n)$, which is compact.
\end{remark}

\begin{proof}
The proof proceeds via a transformation to the interaction picture. For a Schrödinger state $| \varphi_S(t) \rangle$, define the corresponding interaction state 

\begin{equation}\label{eq:StateInter}
    | \varphi_I(t) \rangle = U_{0,S}(t)^\dagger | \varphi_S(t) \rangle.
\end{equation}

For operators, we recall that for a Hermitian operator $A_S$ in the Schrödinger picture, the corresponding operator in the interaction picture is given by

\begin{equation}\label{eq:OpInter}
    A_I = U_{0,S}(t)^\dagger A_S U_{0,S}(t). 
\end{equation}

It follows from the standard theory that time evolution operator $U_I(t)$ satisfies the interaction picture Schrödinger equation

\begin{equation}\label{eq:TE_interaction}
    \frac{d}{dt}U_I(t) = -iH_{1,I}(t) U_I(t),
\end{equation}

where $H_{1,I} = U_{0,S}(t)^\dagger H_{1,S}(t) U_{0,S}(t)$.

Of course, if $H_{1,S}$ is 0, i.e. if there is no noise, then there is no state evolution (and only operators evolve with time). Moreover, since the transformation to the interaction picture is unitary, we have that 

\begin{align}
    \norm{U_S(t, \cdot) - U_{0,S}(t)} = \norm{U_I(t, \cdot) - \id}.
\end{align}

The quantity on the right-hand side is easy to estimate. Fixing $\omega \in \Omega$, we rewrite the interaction picture Schrödinger equation as an integral equation:

\begin{align}
\begin{split}
    &U_I(t, \omega) - \id =- i \int_0^t H_{1, I}(s, \omega)U_I(s, \omega) ds \\
    &= -i \int_0^t U_{0,S}(s)^\dagger H_{1,S}(s, \omega) U_{0,S}(s) U_I(s, \omega)ds \,.
\end{split}
\end{align}

Hence, 

\begin{align}
\begin{split}
    &\norm{U_I(t, \omega)-\id} = \norm{\int_0^t U_{0,S}(s)^\dagger H_{1,S}(s, \omega) U_{0,S}(s) U_I(s, \omega)ds } \\
    &\leq \int_0^t \norm{U_{0,S}(s)^\dagger H_{1,S}(s, \omega) U_{0,S}(s) U_I(s, \omega)}ds \\
    &= \int_0^t \norm{U_{0,S}(s)^\dagger} \norm{H_{1,S}(s, \omega)} \norm{U_{0,S}(s)} \norm {U_I(s, \omega)} ds \\
    &= \int_0^t \norm{H_{1,S}(s, \omega)} ds \leq Kt,\label{eq:basic_error}
\end{split}
\end{align}

where the last inequality holds for almost all $\omega$.

\end{proof}

\begin{remark}
    We have a similar result for any submultiplicative matrix norm, not just the operator norm or Frobenius norm, potentially at the cost of replacing $K$ by a larger constant, since multiplying by a unitary matrix will not preserve arbitrary matrix norms. However, since the only norms we consider in this paper are the operator norm and the Frobenius norm, both of which are preserved by unitary matrices.
\end{remark} 
\subsection{A geometric error estimate}

Unfortunately, the inequality proven in the previous section is not tight, since the left hand side measures the distance between $U_I(t, \omega)$ and the identity in the space of matrices, whereas the integral on the right hand side measures the length of a path on $SU(n)$.  In order to get a tighter bound, we need to work intrinsically on $SU(n)$. 

When working intrinsically on $SU(n)$, it is most natural to use the Frobenius norm $\norm{\cdot}_F$ on matrices, since the associated inner product pulls back via the inclusion $SU(n) \rightarrow M_n(\C)$ to the unique bi-invariant metric on $SU(n)$ that makes the inclusion an isometric embedding. 

Let us assume that $\norm{H_{1, I}(t)}_F \leq K$ almost surely for all $t$. For $U, V \in SU(n)$, let $d_F(U,V)$ denote the distance between $U$ and $V$ on $SU(n)$ measured with respect to the metric induced by the Frobenius norm (also known as the Killing metric).

We have the following result:

\begin{align}
\begin{split}
    \norm{U_I(t, \omega) - \id}_{op}
    &\leq \norm{U_I(t, \omega) - \id}_F \\
    &\leq d_F(U_I(t,\omega), \id) \\
    &\leq \int_0^t \norm{H_{1,I}(s, \omega )U_I(s, \omega)}_Fds \\
    &= \int_0^t \norm{H_{1,I}(s, \omega)}_Fds \leq Kt \quad a.s..
\end{split}
\end{align}

Of course, we may wish to study noise whose amplitude varies with time. Defining \\ $\Lambda(t) := \ess \, \sup_\omega \norm{H_{1,I}(t, \omega)}_F$, the previous is easily adapted to show that

\begin{align}\label{eq:bestbound}
     d_F(U_I(t, \omega), \id) \leq \int_0^t \norm{H_{1,I}(s, \omega)}_Fds \leq \int_0^t \Lambda(s) ds .
\end{align}

We have proven the following theorem:

\geometricerror*

\subsection{Tightness of error bound}

In fact, for certain noise processes, the error bound \eqref{eq:bestbound} is optimal, i.e. the noise can push the error arbitrarily close to the bound given by \eqref{eq:bestbound}. 

\tightness*

\begin{definition}\label{def:WCNC1}
We say that $H_{1,I}(t, \cdot)$ satisfies the worst-case noise condition with respect to $H_{0,S}$ if the following two conditions hold. The first condition is that:

\begin{equation}
    \CP(H_{1,I}(t, \cdot) = 0) = 0
\end{equation}
for all $t$. In this case, we write 

\begin{equation}
    H_{1,I}(t, \omega) = \lambda(t, \omega) \hat{H}_I(t, \omega),
\end{equation} 
where $\norm{\hat{H}_{1,I}(t, \omega)}_F = 1$. 

The second condition is that for all $\epsilon > 0$ and for all $t \geq 0$, 

\begin{align}
\begin{split}
    &\CP[|\lambda(t, \omega) - \Lambda(t) |< \epsilon \\ &\mathrm{and} \norm{\hat{H}_{1,I}(t, \omega) - \hat{H}_{1,I}(0, \omega)}_F < \epsilon] > 0.
\end{split}
\end{align}

\end{definition}

For convenience, we denote
\begin{align}
\begin{split}
    E_\epsilon^I = \{&\omega \in \Omega \; \mathrm{s.t.} \;|\lambda(t, \omega) - \Lambda(t)|< \epsilon \; \\ &\mathrm{and} \; \norm{\hat{H}_{1,I}(t, \omega) - \hat{H}_{1,I}(0, \omega)}_F < \epsilon\}.
\end{split}
\end{align}

If $H_{1,I}(t, \cdot)$ satisfied the WCNC with respect to $H_{0,I}(t,\cdot)$, then we have that $\CP(E^I_\epsilon)>0$ for all $\epsilon > 0$.

Suppose that $E_\epsilon^I$ occurs, that is, $\omega \in \Omega$. Then $H_{1,I}(t, \omega) = \Lambda(t)\hat{H}_{1,I}(0, \omega) + G(t, \omega)$, where $\norm{G(t)}_F$ is $\mathcal{O}(\epsilon)$. 

First, define $\tilde{U}_I(t) = \exp(-i\int_0^t \Lambda(s)\hat{H}_{1,I}(0)ds)$. Note that $\tilde{U}_I(t)$ defines a geodesic on $SU(n)$ and moreover, that $d_F(I, \tilde{U}_I(t)) = \int_0^t \Lambda(s) ds$. Hence, $\tilde{U}_I(t)$ achieves the bound \eqref{eq:bestbound}.

Next, we will show that $d_F(U_I(t, \omega), \tilde{U}_I(t))$ is $\mathcal{O}(\epsilon)$. One relatively straightforward way to do this is via a second interaction picture transformation. Define $\tilde{H}_{1,I}(t) = \Lambda(t)\hat{H}_{1,I}(0, \omega)$. Notice that $\tilde{H}_{1,I}(t)$ commutes with itself at all times, and hence by definition $\tilde{U}_{1,I}(t))$ is the unitary propagator associated to $\tilde{H}_{1,I}(t, \omega)$. Now, perform another transformation to the interaction picture, and denote operators in the double interaction picture by a subscript $I'$. By definition, $H_{1,I'}(t, \omega) = \tilde{U}_I(t, \omega)^\dagger G(t, \omega)\tilde{U}(t, \omega)$, and letting $U_{I'}(t, \omega)$ denote the time evolution operator in the (double) interaction picture, we have the equation

\begin{align}
    \frac{d U_{I'}(t, \omega)}{dt} = - iH_{1,I'}(t, \omega)U_{I'}(t, \omega).
\end{align}

Now, since $\tilde{U}_I(t, \omega)$ is unitary, we have that $\norm{H_{1,I'}(t, \omega)}_F = \mathcal{O}(\epsilon)$, and we also know that $d_F(U_I(t, \omega), \tilde{U}_I(t, \omega)) = d_F(U_{I'}(t, \omega), \id)$. But since $H_{1,I'}(t, \omega)$ is $\mathcal{O}(\epsilon)$, we see that $d_F(U_{I'}(t, \omega), \id)$ is $\mathcal{O}(\epsilon)$, for example by considering the Magnus expansion. 

Putting this all together and using the triangle inequality for $d_F(\cdot,\cdot)$ on $SU(n)$, we see that

\begin{align}
     \int_0^t \Lambda(s) ds  - d_F(I, U_I(t)) < d_F(U_I(t),  \tilde{U}_I(t)) = \mathcal{O}(\epsilon). 
\end{align}

But since by assumption, for all $\epsilon$, $\CP(E_\epsilon^I) > 0$, we conclude that

\begin{align}
    \CP \left[ \int_0^t \Lambda(s) ds - d_F(U_I(t, \omega), \id) < \epsilon \right] > 0
\end{align}

for all $\epsilon$.

\subsection{Deriving a condition on the Schrödinger noise Hamiltonian}

One slightly unsatisfactory part of this reasoning is that the worst case noise condition is defined in terms of the interaction-picture noise Hamiltonian $H_{1,I}(t, \omega)$, which of course depends on the $H_{0,S}$. In the context of control theory, this is problematic. In this section, we derive a sufficient condition on the Schrödinger noise Hamiltonian.

First note that if $H_{1,I}(t, \omega) = 0$ if and only if $H_{1,S}(t, \omega) = 0$, and if $H_{1,I}(t, \omega) = \lambda(t, \omega)\hat{H}_{1,I}(t, \omega)$, then similarly $H_{1,S}(t, \omega) = \lambda(t, \omega) \hat{H}_{1,S}(t, \omega)$. Now note that, by the definition of the interaction picture,

\begin{align*}
\hat{H}_{1,I}(t, \omega) &\in B_{F,\epsilon}(\hat{H}_{1,I}(0, \omega)) \\ &\Updownarrow \\
\hat{H}_{1,S}(t, \omega) &\in  B_{F,\epsilon}(U_{0,S}(t) \hat{H}_{1,S}(0, \omega) U_{0,S}(t)^\dagger).
\end{align*}

The next step is to remove the dependence on the control $U_{0,S}(t)$. Denote by $CT$ the space of all allowable control trajectories $[0,T] \rightarrow SU(n)$, i.e. the space of all trajectories generated by admissible controls $H_{0,S}(t)$. Of course, precisely which controls are admissible depends on the context. 

This motivates the following definition:

\begin{definition}\label{def:WCNC2}
We say that $H_{1,S}(t, \omega)$ satisfies the worst-case noise condition with respect to $CT$ if 
\begin{align}\label{eq:WCNC}
\begin{split}
    \CP[\hat{H}_{1,S}(t, \omega) \in  B_{F,\epsilon}(U_{0,S}(t) \hat{H}_{1,S}(0, \omega) U_{0,S}(t)^\dagger); \\ \; \Lambda(t)-\lambda(t,w) < \epsilon ;\ \forall t ] 
\end{split}
\end{align}
in strictly positive for all $\epsilon$. A simpler way to view this condition is that $H_{1,S}$ satisfies the WCNC with respect to $CT$ if and only if $H_{1,S}$ satisfies the WCNC with respect to all $H_{0,S}(t)$ generating trajectories in $CT$.
\end{definition}

Of course, given a class of admissible control trajectories $CT$, we might ask which noise processes satisfy this condition. In general, this is a difficult question, but we will show that in most cases this class contains several physically reasonable noise processes. 

For bounded controls $H_{0,S}(t)$, there is a sufficient condition to check whether $H_{1,S}(t, \cdot)$ satisfies the WCNC with repect to $H_{0,S}(t)$.

\markov*

\begin{proof}
Suppose $H_{1,S}$ satisfies the assumptions stated above. 

We begin by observing that the WCNC is equivalent to the condition that :

\begin{equation}
\begin{split}
\CP [ H_{1,S}(t, \omega) \in B_{F, \Lambda(t) \epsilon}(&U_{0,S}(t) H_{1,S}(0, \omega)U_{0,S}(t)^\dagger) \; \forall t ] > 0.
\end{split}
\end{equation}

The next step is to find an inner approximation of 

\begin{align}
\begin{split}
    \bigcup_{t \in [0,T]}B_{F, \Lambda(t) \epsilon}(&U_{0,S}(t) H_{1,S}(0, \omega)U_{0,S}(t)^\dagger) \times \{t\} \\ \subset \mathcal{H} \times [0,T]
\end{split}
\end{align}

by a finite concatenation of cylindrical sets. This is always possible, since $\frac{d}{dt} U_{0,S}(t) H_{1,S}(0, \omega)U_{0,S}(t)^\dagger$ is bounded by $E\lambda(0, \omega)$. Precisely, there exists a partition $0 = t_0 < \cdots <t_K=T$ of $[0,T]$ and $\epsilon_i < \epsilon$ such that: 

\begin{align}\label{eq:innerapprox}
\begin{split}
    \bigcup_{i=0}^K \bigcup_{t \in [t_i, t_{i+1}]}B_{F, \Lambda(t_i)\epsilon_i}(U_{0,S}(t_i)H_{1,S}(0, \omega)U_{0,S}(t_i)^\dagger) \times \{t\} \\ \subset \bigcup_{t \in [0,T]}B_{F, \Lambda(t) \epsilon}(U_{0,S}(t) H_{1,S}(0, \omega)U_{0,S}(t)^\dagger) \times \{t\}
\end{split}
\end{align}

and such that

\begin{align}\label{eq:intersection}
\begin{split}
&B_{F, \Lambda(t_i)\epsilon_i}(U_{0,S}(t_i)H_{1,S}(0, \omega)U_0(t_i)^\dagger) \, \\ &\cap \, B_{F, \Lambda(t_{i+1})\epsilon_{i+1}}(U_{0,S}(t_{i+1})H_{1,S}(0, \omega)U_{0,S}(t_{i+1})^\dagger).
\end{split}
\end{align}

is non-empty.

To ease notation, let 
\begin{equation}
    B_i = B_{F, \Lambda(t_i)\epsilon_i}(U_{0,S}(t_i)H_{1,S}(0, \omega)U_{0,S}(t_i)^\dagger)
\end{equation}

for all $i = 0, \dots, K-1$. Now, note that

\begin{align}
\begin{split}
    \CP\left[H_{1,S}(t, \omega)([0,T]) \subset \bigcup_{i=0}^K \bigcup_{t \in [T_i, T_{i+1}]}B_i \times \{t\}\right] \\
    = \CP\left[H_{1,S}(t, \omega) \in B_i\; \text{for} \,t \in [t_i, t_{i+1}] \; \forall i\right].
\end{split}
\end{align}
Denote by $F_i$ the event 

\begin{equation}
    H_{1,S}(t, \omega) \in B_i\; \text{for} \,t \in [t_i, t_{i+1}].
\end{equation}

By the assumption that the path measure associated to $H_{1,S}$ is positive on cylindrical sets, $\CP(E_i)>0$ for all $i$. By the strong Markov assumption, we have that:
\begin{align}
    \CP[F_{i+1}|F_i] = \CP[F_{i+1}|H_{1,S}(t_i, \omega) \in B_{i+1}\cap B_i].
\end{align}
Now, for all $i$, $\CP[H_{1,S}(t_i, \omega) \in B_{i+1}\cap B_i]]$ is non-zero. By the definition of conditional probability, we conclude that
\begin{align}
    \CP\left[H_{1,S}(t, \omega) \in B_i\; \text{for} \,t \in [t_i, t_{i+1}] \; \forall i\right] > 0.
\end{align}
But since 
\begin{align}
\begin{split}
    &\bigcup_{i=0}^K \bigcup_{t \in [t_i, t_{i+1}]}B_i \times \{t\} \\ &\subset \bigcup_{t \in [0,T]}B_{F, \Lambda(t) \epsilon}(U_{0,S}(t) H_{1,S}(0, \omega)U_{0,S}(t)^\dagger) \times \{t\} \\ &\subset \mathcal{H} \times [0,T]
\end{split}
\end{align}
the result follows. 
\end{proof}

\begin{remark}\label{rk:nonemptycondition}
Identify the space of Hermitian matrices with $\R^k$ for the appropriate $k$ Suppose that $X(t, \cdot)$ is an $k$-dimensional unbounded Markov process, for example a Wiener process, an Ornstein-Uhlenbeck process, or a geometric Brownian motion. Let $f:\R^k \rightarrow B_E(0)$ be a continuous bijection. Then it is a standard fact from stochastic analysis that $f(X_t)$ is still a Markov process. Moreover, it is straightforward to check that if the path measure associated $X_t$ is positive on cylindrical sets in $[0,T]\times \R^k$, then the path measure associated to $f(X_t)$ is positive on cylindrical sets in $[0,T] \times \mathcal{H}_E \subset [0,T] \times \R^k$. 

Hence, for example, in one dimension, letting $W(t, \cdot)$ denote the Wiener process, the random process
\begin{align}
    \frac{2E}{\pi} \arctan(W(t, \cdot))
\end{align}
satisfies the assumptions of \cref{thm:markov} and hence the WCNC. 

We leave it to the reader to come up with an $n-$dimensional example by applying their favourite continuous bijection $\R^n \leftrightarrow B_E(0)$ to their favourite Markov process that satisfies the positivity condition on cylindrical sets. As a heuristic, one expects that most Markov processes with almost surely continuous sample paths will satisfy this last condition. 
\end{remark}

\section{Robust Quantum Control}

\subsection{Control noise-couplings}

In this final section, we consider the problem of synthesising a target unitary $V$ while minimising the worst-case error $\ess\,\sup_\omega\norm{U_S(t, \omega) - V}$, in the presence of a control noise coupling $H_{1,S}(t, \omega) = H_{1,S}(H_{0,S}(t),t, \omega)$. We will restrict our attention to optimising over control Hamiltonians $H_{0,S}(t)$ generating $V$ exactly in the absence of noise. 

\begin{remark}
    It would be tempting to talk about control Hamiltonians $H_{0,S}$ generating $V$ in expectation. But this would be incorrect, unless the expectation is taken to be a \textit{Fréchet mean}.
\end{remark}

There are many natural couplings one could consider. The most comprehensive treatment of control noise couplings is given in \cite{Kallush_2014} and \cite{muller_et_al}, in the context of the stochastic Schrödinger equations, though the ideas easily translate to the RODE setting. 

In each of these papers, the amplitude of the the noise is assumed to be proportional to the intensity of the control function, and in \cite{muller_et_al}, the authors allow for both multiple noise terms and multiple control terms. It is the latter approach that we consider here. 

We are interested in situations where moving in certain unfavourable Hamiltonian directions induces more noise than moving in more favourable directions. For example, in the multi-qubit setting, we generally expect that applying a Hamiltonian  corresponding to multiple qubit operations will induce more noise than applying a Hamiltonian corresponding to one and two-qubit interactions, as in \cite{dowling2008geometry, nielsen2006geometric,nielsen2006optimal, nielsen2006quantum}. These couplings are natural even in the single qubit case. For example, when controlling a single spin qubit, \cite{PhysRevD.100.046020} suggests a situation where a magnetic field can be easily applied in the $x-y$ plane, but only applied with more difficulty in the $z$ axis. 

\subsection{Couplings via noise metrics}

We would like a coupling that captures both the proportionality of the amplitude of the noise to the intensity of the control functions, while also allowing for a ``directional dependency". Fix an orthonormal global frame $\{H_j\}$. Then any control field can be written as a sum $H_{0,S} = \sum_j h_j(t) H_j$. We will consider couplings $H_{1,S}(t, \omega) = H_{1,S}(H_{0,S}, t, \omega)$ such that

\begin{align}
\begin{split}
    \ess \, \sup_\omega \norm{H_{1,S}(t, \omega)}_F = \Lambda(t) =\Lambda(H_{0,S}(t))\\ = \sqrt{ \sum_j l_j h_j(t)^2}.
\end{split}
\end{align}

If, for example, moving in the $H_j$ direction induces $a$-times as much noise as moving in the $H_k$ direction, then we should set $\Lambda_k = a^2 \Lambda_j$, and so on. Of course, this means that $\Lambda(H_{0,S}(t))$ is actually induced by a Riemannian metric $g_\Lambda$ on $SU(n)$. With respect to the orthonormal global frame $\{H_j\}$, $g_\Lambda$ is represented by the diagonal matrix

\begin{align}
    \begin{pmatrix}l_1 & & \\ & \ddots & \\ & & l_m\end{pmatrix}.
\end{align}

We call such a metric a \textit{noise metric}. Couplings arising from such a metric satisfy both the property of proportionality of the amplitude of the noise to the intensity of the signal, as well as allowing for directional dependence. Moreover, as alluded to in \cite{muller_et_al}, the components of $g_\Lambda$ can be either calculated or inferred experimentally.  

Next, recall that by \cref{thm:geometricerror}, we know that

\begin{align}
    \norm{U_S(t, \omega)-V}_F \leq  d_F(U_S(t, \omega), V) \leq \int_0^t  \norm{H_1(s, \omega)}_Fds
\end{align}

for almsot all $\omega$.

But since, by definition, $\norm{H_{1,S}(t, \omega)} \leq \Lambda(t) = \sqrt{g_\Lambda(H_{0,S}(t), H_{0,S}(t))}$, we see that 

\begin{equation}
\begin{split}
     \norm{U_S(t, \omega)-V}_F \leq \int_0^t\sqrt{g_\Lambda(H_{0,S}(s), H_{0,S}(s))}ds\\ = l_{g_\Lambda}(\gamma)
     \label{eq:robust}
\end{split}
\end{equation}

where $l_{g_\Lambda}(\gamma)$ is the length, with respect to $g_\Lambda$, of the path on $SU(n)$ generated by the control field $H_{0,S}$.

This would suggest that minimising the error 

\begin{align}
    \ess\,\sup_\omega \norm{U_S(t, \omega)-V}_F
\end{align}

is equivalent to finding a length-minimising $g_\Lambda$-geodesic from the identity to $V$. It turns out that this is true, as we might expect from \cref{thm:tightness}. However, we cannot apply \cref{thm:tightness} directly, due to the presence of the control-noise coupling. This motivates the following definition:

\begin{definition}\label{def:WCNC3}
We say that $H_{1,S}(H_{0,S}(t), t, \omega)$ satisfies the \textit{control-dependent worst-case noise condition} if, for each $U_0 \in CT$ generated by a control Hamiltonian $H_{0,S}(t)$, $H_{1,S}(H_{0,S}(t), t, \omega)$ satisfies the WCNC with respect to $H_0(t)$, as in \cref{def:WCNC1}.
\end{definition}

\begin{remark}
We observe that if the noise is independent of the control, we recover the WCNC with respect to $CT$.     
\end{remark} 


With these definitions in place, we can state the main theorem of this section. 

\geodesic*

\begin{proof} First note that minimising $d_F(U_S, V)$ minimises  $\norm{U_S(t, \omega)-V}_F$ The result then follows by applying \cref{thm:tightness} to the Schrödinger RODE for each $H_{0,S}(t)$.
\end{proof}

\begin{figure}[H]
  \centering

  \begin{minipage}{0.25\textwidth}
    \includegraphics[width=\linewidth]{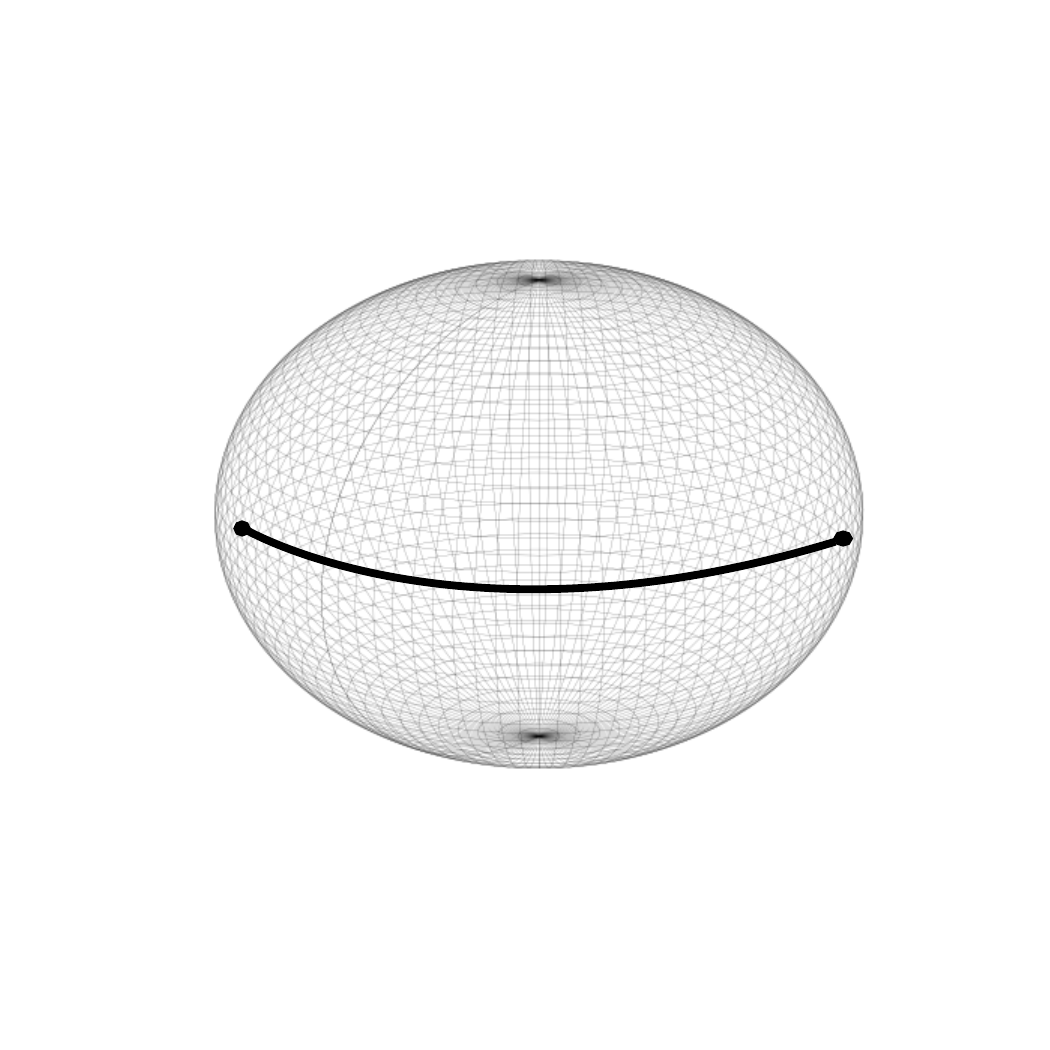}
  \end{minipage}

  \begin{minipage}{0.25\textwidth}
    \includegraphics[width=\linewidth]{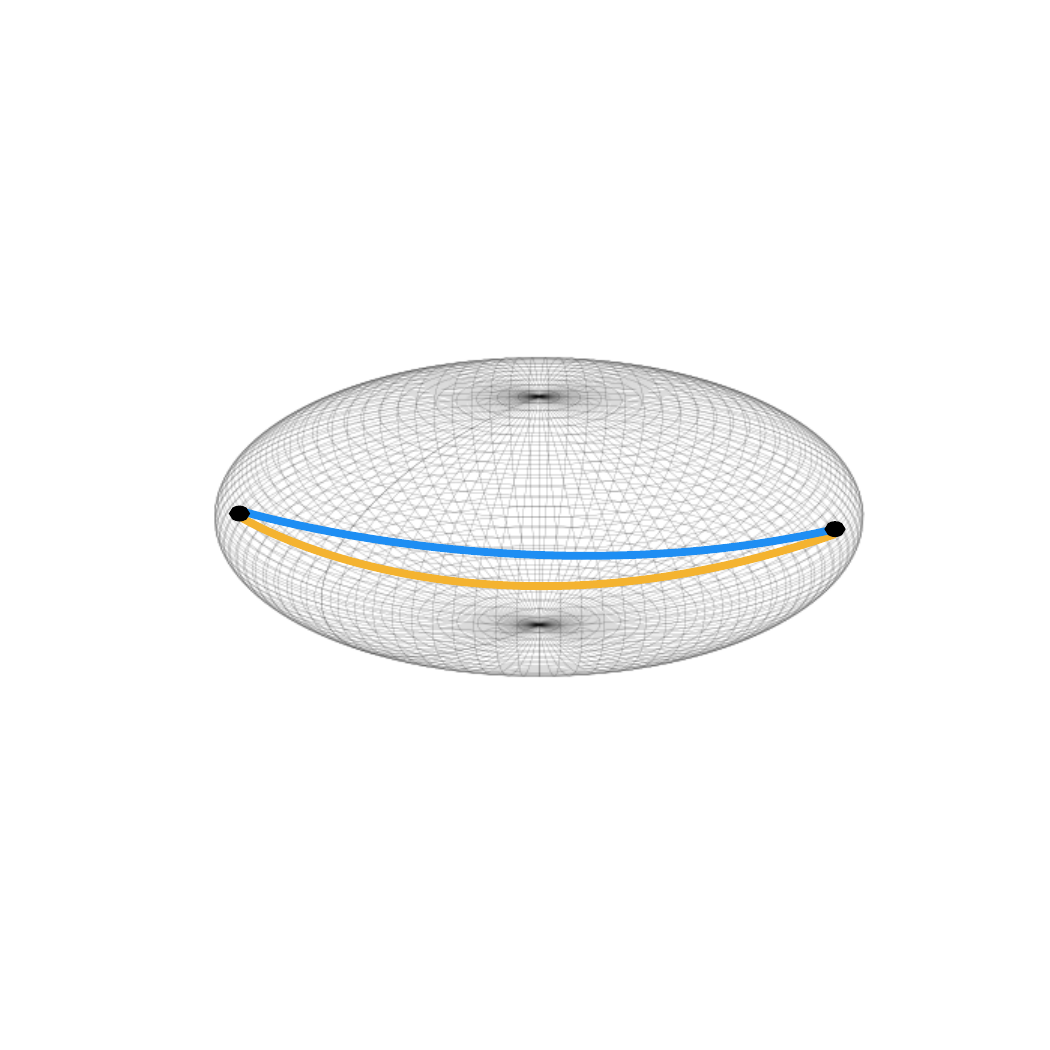}
  \end{minipage}
  
  \begin{minipage}{0.25\textwidth}
    \includegraphics[width=\linewidth]{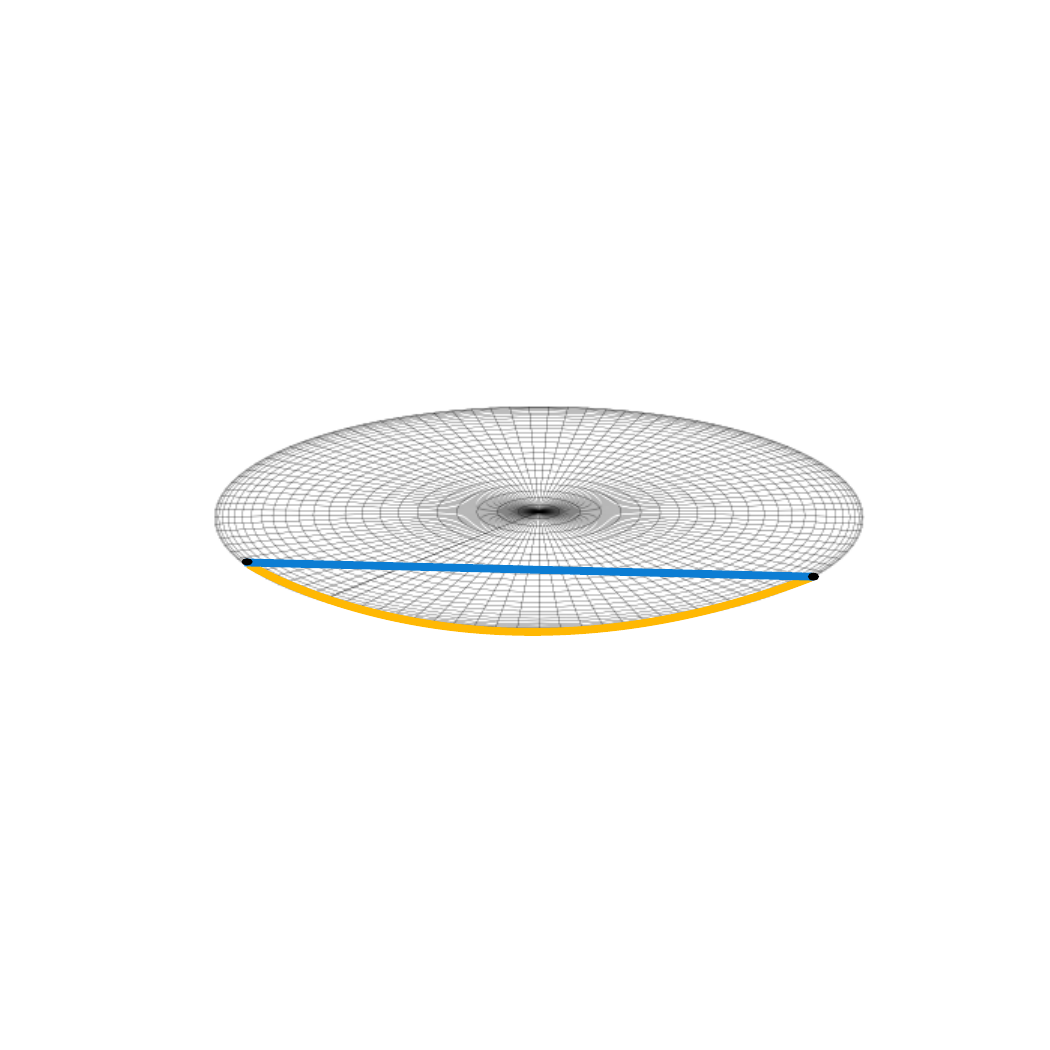}
  \end{minipage}

  \caption{Schematic of a deformation of $SU(n)$ corresponding to penalising one direction. The picture on the left represents (a hyperplane section of) $SU(n)$ with the usual round metric, showing two ``equatorial" points joined by a geodesic. The middle picture shows a deformed $SU(2)$, with the blue line representing the geodesic with respect to the warped metric, and the orange line shows the image of the geodesic with respect to the round metric. The right hand picture shows the limiting case, where one direction has been ``infinitely penalised" effectively reducing the number of tangent directions (and hence the dimension) by one. Note that this diagram is merely a schematic: each $g_\Lambda$ is a homogeneous metric, meaning the curvature should be identical at each point. Unfortunately, it is not possible to satisfactorily depict this in a sketch. For a more detailed discussion, see \cite{PhysRevD.100.046020}.
  \label{fig:deformed}}
\end{figure}

\subsection{Calculating $g_\Lambda$-geodesics}

In theory, calculating geodesics with respect to  $g_\Lambda$ connecting two points $U$ and $V$ on $SU(n)$ is straightforward. We simply need to calculate the Christoffel symbols $\Gamma^a_{bc}$ of the Levi-Civita connection associated to $g_\Lambda$, and then pass the geodesic equation

\begin{align}
    &\ddot{x}^a + \Gamma^a_{bc}\dot{x}^b\dot{x}^c = 0 \\
    &x(0) = U, \; x(1) = V
\end{align}

where $x(t)$ denotes the geodesic path, to an ODE solver capable of handling Dirichlet boundary conditions (for instance via shooting methods). Now, although the geodesic equation is formally a second order equation, it is actually a first order equation in $\dot{x}$. This means that it is tempting to calculate the Christoffel symbols with respect to some orthogonal (or indeed orthonormal) frame of vector fields $\{e_a\}$ in $SU(n)$. With respect to such a basis, and writing $g_\Lambda = g$ to simplify notation, the Christoffel symbols are given by

\begin{align}
\Gamma^a_{bc} = \frac{1}{2}g^{da}\big[e_b(g_{cd}) + e_c(g_{db}) - e_d(g_{bc}) + C_{dbc} + C_{cdb}+ C_{bdc}\big]. 
\end{align}

Now, by the homogeneity of $g$, the first three terms vanish, since $g$ is constant with respect to this frame. Moroever, taking $\{e_a\}$ to be (normalised) generalised Pauli vector fields, computing the structure constants is straightforward. However, it turns out that this approach does not work for our purposes: while we can easily obtain the geodesic equation for $\dot{x} = \dot{x}^ae_a$, we cannot impose Dirichlet boundary conditions $x(0) = U$, $x(1) = V$ in these coordinates. 

If we want to impose Dirichlet boundary conditions, we need to work with respect to a coordinate basis. Working in a coordinate basis still leads to some simplifications, since the structure constants $C_{abc}$ vanish due to symmetry of mixed partials. However, expressing $g_{ab}$ with respect to coordinates on $SU(n)$ is delicate, since we in general expect to have a description of $g = g_\Lambda$ with respect to a basis $\mathfrak{su}(n) = TSU(n)$ of generalised Pauli matrices. 

Recall that the exponential map $\mathfrak{su}(n) \rightarrow SU(n)$ is a local diffeomorphism, and since $\mathfrak{su}(n) \cong \R^{n^2-1}$, the local inverse of the exponential map, corresponding a choice of logarithm $\xi = log\,U$ gives local coordinates on $SU(n)$. 

Now, in these coordinates, standard results from Lie theory, see for example section 6.2 of \cite{toronto_notes} allow us to derive expressions for left-invariant vector fields on $SU(n)$ with respect to these coordinates. Precisely, given $\eta \in \mathfrak{su}(n)$, there is an associated left-invariant invariant vector field $\eta^L$ on $SU(n)$. On an open set $A \subset \mathfrak{su}(n)$ such that $exp:A \rightarrow SU(n)$ is a diffeomorphism onto its image, the coordinate representation is given by pulling back $\eta^L$ by the exponential map. We then have a formula (see e.g. appendix C.4. of \cite{meinrenken}) in terms of the adjoint homomorphism:

\begin{align}
    (exp|_A)^*(\eta^L)|_\xi = \frac{ad_\xi}{1 - exp(-ad_\xi)}\eta.
\end{align}

Now, note that with respect to the coordinates, which we denote $\xi^a$, we have that $g_{ab} = g(\frac{\partial}{\partial \xi^a}, \frac{\partial}{\partial \xi^b})$, and so we can express $g$ with respect to these vector fields, and hence compute (approximately) the Christoffel symbols. 
In the context of optimising one, two and three qubit entangling gates, this approach is tractable, but  it becomes  impractical for large systems. For a $k$-qubit system, we need to perform this procedure on $SU(2^k)$, whose dimension grows exponentially with respect to the number of qubits. For this reason, numerical solution of this equation is impractical for large numbers of qubits without additional assumptions. 

When working in this approach, we also need to contend with two sources of numerical error: clearly, solving the geodesic equation via numerical methods will induce some numerical error, but perhaps more fundamentally, in order to write down the geodesic equation itself, we are obliged to truncate a power series expression for the generalised Pauli-vector fields.

Nevertheless, for larger systems, Nielsen and Dowling \cite{nielsen2006geometric} present an alternative method for calculating geodesics, based on the \textit{lifted Hamilton-Jacobi-Bellman equation} and the notion of a \textit{geodesic derivative}. The basic idea is to compute the geodesics for the round (Killing) metric, which correspond to time-independent Hamiltonians, and can hence be computed by finding a matrix logarithm of the target gate $V$, and then deforming these geodesics with to geodesics of $g_\Lambda$. 

\subsection{The sub-Riemannian limit}

Of course, even if one could solve the geodesic equation efficiently and accurately for a large number of qubits, the resulting fidelity-optimal Hamiltonian would in general contain many high-weight generalised Pauli matrices that are impracticable or indeed impossible to apply to any physical realisation of the system. In most physical systems, we can only apply certain Hamiltonians that lie in the span of a subbundle $B \hookrightarrow TSU(n)$ of the tangent bundle spanned by certain low-weight Pauli matrices. It is a classical theorem in control theory and sub-Riemannian geometry, due to Chow and Rashevskii, that if this sub-bundle generates the whole tangent bundle under the Lie bracket, then any target can be reached by a path whose tangent lies in the subbundle. The subbundle $B$ is said to be \textit{bracket-generating}, or \textit{totally non-integrable} (by contrast with subbundles which are closed under the Lie bracket, which are referred to as integrable). 

In this context, it makes sense to look for length-minimising curves between the identity and the target whose tangent vector lies in $B$. This gives rise to the so-called Carnot-Carathéodory metric: 

\begin{align}
    d(U,V) = inf_\gamma \int_0^1\sqrt{g(\gamma(t),\gamma(t)}\,dt \\
    \dot{\gamma}(t) \in B, \quad \gamma(0) = U, \quad \gamma(1) = V.
\end{align}

The Carnot-Carathéodory metric on $SU(n)$ can be thought of as the limit of a noise metric when certain Hamiltonian directions are infinitely penalised. In \cite{subRiemannian}, the authors make this claim rigorous. The work has important numerical and practical (hardware) implications. 

From a numerical perspective, computing the sub-Riemannian geodesics (i.e. minimisers of the Carnot-Carathéodory distance) should in general be easier than computing the geodesics of the noise metrics introduced here (or, equivalently, of the complexity metrics appearing in \cite{nielsen2006geometric}). This is because sub-Riemannian geodesics are ``sparse" in the following sense: if $B$ denotes the bundle of one and two qubit entangling Hamiltonians (corresponding to weight one and two generalised Pauli matrices) on $k$ qubits, then we need to compute only $\mathcal{O}(k^2)$ coefficients, since there are only $\begin{pmatrix}
    k \\ 2
\end{pmatrix} + \begin{pmatrix}
    k \\ 1
\end{pmatrix}$ generalised Pauli matrices of weight one or two. On the other hand, computing the geodesics of noise metric requires us to compute $2^k$ coefficients. 

From a practical perspective, the sub-Riemannian approach also has a significant advantage, since applying a Hamiltonian to a large number of qubits at once is either impractical or impossible. Computing sub-Riemannian geodesics allows one to restrict the optimal control problem to the space of Hamiltonians that are actually available to to the controller, removing the need for the idealised assumption that the controller is able to apply any Hamiltonian to the system. Of course, the sub-Riemannian approach and the penalty factor approach can be applied at the same time in a natural way. Letting $B$ denote the subbundle of available Hamiltonians, one can deform the metric on $B$ to penalise some of the available directions. This allows us to distinguish ``good", ``bad" and ``impossible" directions in the control problem.

\section{Noise-informed quantum control without calculating geodesics}

It turns out that we can use the insight provided by theorem \ref{thm:tightness} even when we cannot explicitly compute geodesics, using standard techniques from the theory of quantum optimal control. The basic idea is to minimise the error due to the noise by adding a soft constraint in the cost function of the quantum optimal control problem. We refer to this approach noise-informed quantum control.  Moreover, in this setting we can easily incorporate a drift Hamiltonian.

Given $\{H_j\} \in \mathfrak{su}(n)$ and a target $V \in SU(n)$, we seek control functions $u_j(t):[0,T] \rightarrow \R$ solving the robust optimal control problem 

\begin{align}
\begin{split}
    &\mathrm{min}\;c(U(T, \cdot), V) \\
    &\text{subject to} \,\frac{dU(t, \omega)}{dt} = -i(H_d + \sum_j u_j(t)H_j)U(t, \omega).
\end{split}
\end{align}

Here, $H_d$ denotes the drift Hamiltonian of the system, and the cost function $c(U(T),V)$ is generally some distance function on $SU(n)$. Common choices include the operator norm of the difference, the Frobenius norm  of the distance, or their squares. Another common choice is to choose $C(U(T), V) = - \mathrm{Tr}(U(T)^\dagger V)$, such that solving the optimisation problem maximises the \textit{unitary fidelity}. 

In this paper, we are interested in the quantity $\ess\,\sup_\omega\norm{U_S(T, \omega) - V}$. In this section, we no longer consider controls generating $V$ in (Fréchet) expectation, i.e. we don't require that $U_{0,S}(t) = V$. 

Instead, we apply the triangle inequality:

\begin{align}
\begin{split}
    &\ess\,\sup_\omega\norm{U_S(T, \omega) - V} \\
    &\leq \norm{U_{0,S}(T) - V} + \ess\,\sup_\omega\norm{U_S(T, \omega) - U_{0,S}(T)} \\
    & \leq \norm{U_{0,S}(T) - V} + \ess\,\sup_\omega\int_0^T \norm{H_{1,S}(t, \omega)}dt \\
    &=  \norm{U_{0,S}(T) - V} + \int_0^t\sqrt{g_\Lambda(H_{0,S}(t)), g_\Lambda(H_{0,S}(t))} \, dt.
\end{split}
\end{align}

In particular, the first term captures the error due to only being able to solve the deterministic control problem approximately, and the second term corresponds to the error induced by the noise. 

Now, we no longer have an explicit dependence on $\omega$, so we have successfully recast the stochastic (robust) control problem as a deterministic control problem, sometimes called a \textit{robust counterpart}. The cost function is given by:

\begin{align}
\begin{split}
    &c_1(H_0(\cdot)) \\ &= \norm{U_{0,S}(T) - V} + \int_0^t\sqrt{g_\Lambda(H_{0,S}(t)), g_\Lambda(H_{0,S}(t))} \, dt.
\end{split}
\end{align}

In fact, minimising the sum of the squares gives a cost function that is more amenable to standard numerical methods:

\begin{align}
\begin{split}
    &c_2(H_0(\cdot)) \\ &= \norm{U_{0,S}(T) - V}^2 + \Big[\int_0^t\sqrt{g_\Lambda(H_{0,S}(t)), g_\Lambda(H_{0,S}(t))} \, dt]\Big]^2.
\end{split}
\label{eq:75}
\end{align}

Minimising $c_2$ yields a minimiser that provides an upper bound on the true minimum of $c_1$, and provides satisfactory numerical results. 

\subsection{A one-qubit example}

To illustrate the impact of the approximation \eqref{eq:75}, compared to not ignoring noise altogether,  
let us consider the control of a single qubit, with Hamiltonian

\begin{align}
    H = H_d + \sum_{j=1}^3 u_j \sigma_j + H_n(t, \omega).
\end{align}

Here, $H_d$ denotes the drift Hamiltonian, which we take to be $\frac{1}{2}\sigma_y$. The space of admissible control functions is given by fifth-degree polynomials of the form $u_j= \sum_{k=0}^5c_jt^k$, which we initialise by choosing the $c_k$ randomly.

We assume that the coupling between $H_n$ is given by the $g_\Lambda$ from the previous example, 

\begin{align}
    g_\Lambda = \begin{pmatrix}1 & & \\ & \frac{1}{100} & \\ & & \frac{1}{100}\end{pmatrix},
\end{align}

so applying $\sigma_x$ induces ten times as much noise as applying $\sigma_y$ or $\sigma_z$. 

We now run two numerical experiments \footnote{Code available at \url{https://github.com/llorebaga/RandomSE}}, where we solve the quantum control problem in \textit{Julia} \cite{Julia-2017} using the numerical optimization package \textit{Optim} \cite{Optim}. In the first one, we minimize the cost function given by
\begin{align}\label{eq:nonoisecost}
    \norm{U_{0,S}(T) - V}^2, 
\end{align}
which does not take the noise into account. In the second, we perform noise informed optimization and minimize
\begin{align}\label{eq:noisecost}
\begin{split}
    &c_2(H_0(\cdot)) \\ &= \norm{U_{0,S}(T) - V}^2 + \Big[\int_0^t\sqrt{g_\Lambda(H_{0,S}(t)), g_\Lambda(H_{0,S}(t))} \, dt\Big]^2.
\end{split}
\end{align}

For this, we sample from a noise process in a way that satisfies the coupling. We simulate 100 noise trajectories, setting the entries of $H_{1,S}(t, \cdot)$ to the coordinate-wise arctangent of a Gaussian process with Matérn covariance. We then multiply by the envelope function $\frac{1}{2 \pi} g_\Lambda(H_0(t), H_0(t))$, which  guarantees that $\ess \, \sup \,H_n(t) = g_\Lambda(H_0(t), H_0(t))$.

Finally, we evolve the system with each control (noise-aware and noise-blind) and compare the difference in fidelity, as shown in a violin plot in Figure \ref{fig:optimal_control}. 
 
\begin{figure}[t]
    \centering
    \includegraphics[width=\linewidth]{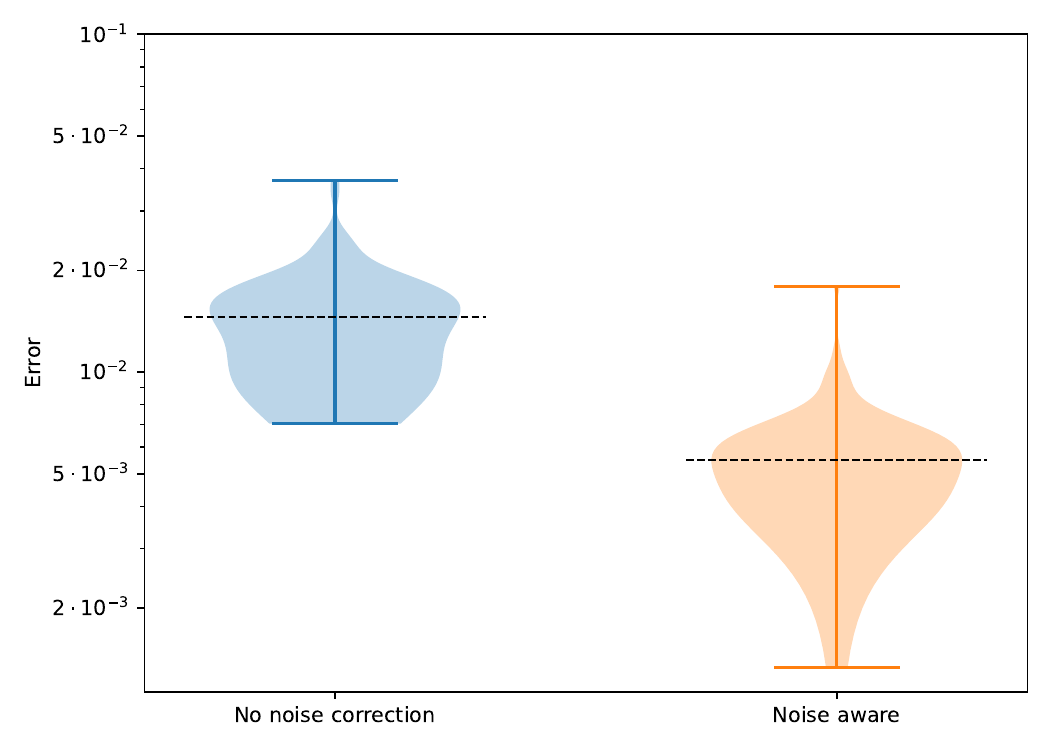}
  \caption{Violin plots of the error, accounting for the noise (i.e. minimising \eqref{eq:noisecost}) and not accounting for the noise (i.e. minimsing \eqref{eq:nonoisecost}). Each distribution consists of 100 random samples for the control functions, with the dashed lines indicating the mean. We see an approximately 2.6-fold reduction in operator-norm error.}
  \label{fig:optimal_control}
\end{figure}

\begin{figure}[t]
    \centering
    \includegraphics[width=\linewidth]{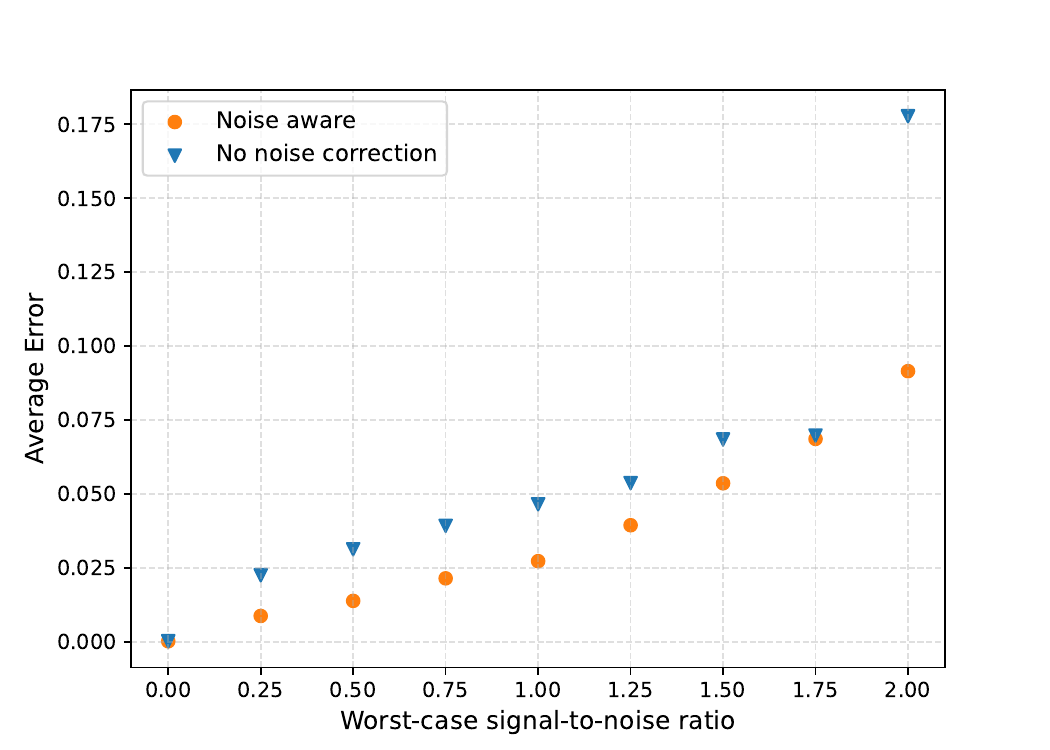}
  \caption{A plot of the average unitary error accounting for the noise (orange circles) and not accounting for the noise (blue triangles), for 100 random control functions.}
  \label{fig:ratio}
\end{figure}

In this example, the worst-case signal-to-noise ratio, $\eta$, corresponding to the largest entry of $g_\Lambda$, is one. However, it turns out that the reduction in unitary error is robust across a fairly large range of signal-to-noise ratios. To study this, we consider signal-to-noise ratios varying in $[0,2.0]$, where 0 implies no noise and $2.0$ implies a major impact of noise as the system evolves, as we see in the noise term
\begin{align}\label{eq:noisecost}
\begin{split}
    &\int_0^t \eta \, \sqrt{g_\Lambda(H_{0,S}(t)), g_\Lambda(H_{0,S}(t))} \, dt.
\end{split}
\end{align}
In Figure \ref{fig:ratio} we follow a similar setting as for Figure \ref{fig:optimal_control}, but we the signal-to-noise ratio $\eta$. We plot the average unitary error, both in the noise aware setting, and ignoring the noise.

Both figures \ref{fig:optimal_control} and \ref{fig:ratio} illustrate how adding a noise term in the optimization problem helps find a more robust solution. We observe a linear increase in the average error with the signal-to-noise ratio, and a more regular behavior from the noise-aware results. 

\section{Conclusion}

Using the theory of random ordinary differential equations, we have proposed a new model of open quantum systems describing a large class of quantum systems subject to external noise. In particular, we prove tight bounds on the error induced by the noise. Inspired by the work of Nielsen, we have introduced a natural class of couplings between the control field and the noise term, which are described by a left-invariant metric on $SU(n)$ called a noise metric. 

Additionally, we have demonstrated an approach to synthesising a target unitary while minimising the worst-case error, which can be construed as robust quantum control. By minimising the worst-case error, we are able to transform a stochastic optimal control problem into a more tractable deterministic control problem. This is similar to the approach of constructing a \textit{robust counterpart} in optimisation theory. When the noise is coupled to the control via a noise metric, the optimal trajectory is a geodesic on $SU(n)$ with respect to the noise metric. 

More work is needed to explicitly compute these geodesics, particularly in higher dimension, and in the sub-Riemannian limit. We hope to address these questions in future work. We have presented an alternative approach, referred to as noise-informed quantum optimal control, which does not involve computing geodesics. Instead, the method accounts for noise by incorporating a term corresponding to the noise as a soft constraint in classical optimal control problem, which can be solved using standard methods.

\section{Acknowledgments}


The work of Rufus Lawrence, Aleš Wodecki, Johannes Aspman and Jakub Mareček has been supported by the Czech Science Foundation (23-07947S).
Llorenç Balada Gaggioli has been supported by
European Union’s HORIZON–MSCA-2023-DN-JD programme under under the Horizon Europe (HORIZON) Marie Sklodowska-Curie Actions, grant agreement
101120296 (TENORS).

\end{document}